\documentclass[11pt, runningheads]{llncs}

\usepackage{amsmath}
\usepackage{graphicx}
\usepackage{subfigure}
\usepackage{algorithm}
\usepackage{algorithmic}

\renewcommand{\algorithmicrequire}{\textbf{Input:}}
\renewcommand{\algorithmicensure}{\textbf{Output:}}
\renewcommand{\algorithmicprint}{\textbf{execute}}

\newcommand{\removed}[1]{}

\title{Passively Mobile Communicating Logarithmic Space Machines\thanks{This work has been partially supported by the ICT Programme of the European Union under contract number ICT-2008-215270 (\textsf{FRONTS}).}}

\author{Ioannis Chatzigiannakis\inst{1,2} \and Othon Michail\inst{1,2} \and Stavros Nikolaou\inst{2} \and \\ Andreas Pavlogiannis\inst{2} \and Paul G. Spirakis\inst{1,2}}
\institute{Research Academic Computer Technology Institute (RACTI), Patras, Greece\and Computer Engineering and Informatics Department (CEID), University of Patras, Patras, Greece.\\
Email:\email{ \{ichatz, michailo, spirakis\}@cti.gr, \{snikolaou, paulogiann\}@ceid.upatras.gr}}

\titlerunning{Passively Mobile Communicating Logarithmic Space Machines}

\begin{document}

\maketitle

\begin{abstract}
We propose a new theoretical model for passively mobile Wireless Sensor Networks. We call it the \emph{PALOMA} model, standing for PAssively mobile LOgarithmic space MAchines. The main modification w.r.t. the Population Protocol model \cite{AADFP06} is that agents now, instead of being automata, are Turing Machines whose memory is logarithmic in the population size $n$. Note that the new model is still \emph{easily implementable} with current technology. We focus on \emph{complete communication graphs}. We define the complexity class $PLM$, consisting of all symmetric predicates on input assignments that are stably computable by the PALOMA model. We assume that the agents are initially \emph{identical}. Surprisingly, it turns out that the PALOMA model \emph{can assign unique consecutive ids to the agents and inform them of the population size}! This allows us to give a direct simulation of a \emph{Deterministic} Turing Machine of $\mathcal{O}(n\log n)$ space, thus, establishing that any symmetric predicate in $SPACE(n\log n)$ also belongs to $PLM$. We next prove that the PALOMA model can simulate the Community Protocol model \cite{GR09}, thus, improving the previous lower bound to all symmetric predicates in $NSPACE(n\log n)$. Going one step further, we generalize the simulation of the deterministic TM to prove that the PALOMA model can simulate a \emph{Nondeterministic} TM of $\mathcal{O}(n\log n)$ space. Although providing the same lower bound, the important remark here is that the bound is now obtained in a \emph{direct manner}, in the sense that \emph{it does not depend on the simulation of a TM by a Pointer Machine}. Finally, by showing that a Nondeterministic TM of $\mathcal{O}(n\log n)$ space decides any language stably computable by the PALOMA model, we end up with an exact characterization for $PLM$: \emph{it is precisely the class of all symmetric predicates in} $NSPACE(n\log n)$.
\end{abstract}

%\newpage
%\setcounter{page}{1}

%\noindent
%\textbf{Note: Submitted to Track C.} Due to space restrictions, we have moved some proofs and some %additional material to a clearly marked Appendix, to be read at the discretion of the PC.

\section{Introduction}

Theoretical models for Wireless Sensor Networks have received great attention over the past few years. Recently, Angluin et al. \cite{AADFP06} proposed the \emph{Population Protocol} (\emph{PP}) model. Their aim was to model sensor networks consisting of tiny computational devices with sensing capabilities that follow some unpredictable and uncontrollable mobility pattern. Due to the minimalistic nature of their model, the class of computable predicates was proven to be fairly small: it is the class of \emph{semilinear predicates}, thus, not including multiplication of variables, exponentiations, and many other important operations on input variables. Moreover, Delporte-Gallet et al. \cite{DFGR06} showed that PPs can tolerate only $\mathcal{O}(1)$ crash failures and not even a single Byzantine agent.

The work of Angluin et al. shed light and opened the way towards a brand new and very promising direction. The lack of control over the interaction pattern, as well as its inherent nondeterminism, gave rise to a variety of new theoretical models for WSNs. Those models draw most of their beauty precisely from their inability to organize interactions in a convenient and predetermined way. In fact, the Population Protocol model was the minimalistic starting-point of this area of research. Most efforts are now towards strengthening the model of Angluin et al. with extra realistic and implementable assumptions, in order to gain more computational power and/or speed-up the time to convergence and/or improve fault-tolerance.

The \emph{Mediated Population Protocol} (\emph{MPP}) model of \cite{CMS09-2} was based on the assumption that the agents can read and write pairwise information on some global storage, like e.g. a base station. Each ordered pair of agents that is permitted to interact occupies a memory slot of fixed length. The MPP model is strictly stronger than the PP model, since it can handle multiplication of variables, and \emph{MP}, which is the class of symmetric predicates stably computable by MPP, is contained in $NSPACE(m)$, where $m$ denotes the number of edges of the communication graph. Unfortunately, as partly indicated by \cite{CMS09-3} and \cite{CMS09-4}, our knowledge of this model is quite restricted yet.

Guerraoui and Ruppert \cite{GR09} made another natural assumption: they equipped the agents with read-only unique identifiers picked from an infinite set of ids. They named their model the \emph{Community Protocol} model. Each agent has its own unique id and can store up to a constant number of other agents' ids. In this model, agents are only allowed to compare ids, that is, no other operation on ids is permitted. It was proven that the corresponding class consists of all symmetric predicates in $NSPACE(n\log n)$, where $n$ is the community size. The proof was based on a simulation of a modified version of Nondeterministic Sch\" onhage's \emph{Storage Modification Machine} (\emph{NSMM}). It was additionally shown that if faults cannot alter the unique ids and if some necessary preconditions are satisfied, then community protocols can tolerate $\mathcal{O}(1)$ Byzantine agents.

In this work, we think of each agent as being a Turing Machine whose memory is logarithmic in the population size. Based on this realistic and implementable assumption, we propose a new theoretical model for passively mobile sensor networks, called the \emph{PALOMA} model. To be more precise, it is a model of PAssively mobile MAchines (that we keep calling agents) with sensing capabilities, equipped with two-way communication, and each having a memory of size LOgarithmic in the population size $n$. The reason for studying such an extension is that having \emph{logarithmic communicating machines} seems to be more natural than communicating automata of constant memory. First of all, the \emph{communicating machines} assumption is perfectly consistent with current technology (cellphones, iPods, PDAs, and so on). Moreover, \emph{logarithmic} is, in fact, \emph{extremely small}. For a convincing example, it suffices to mention that for a population consisting of $2^{266}$ agents, which is a number greater than the number of atoms in the observable universe, we only require each agent to have $266$ cells of memory (while small-sized flash memory cards nowadays exceed $16$GB of storage capacity)! Interestingly, as we shall see, it turns out that the agents are able to organize themselves into a distributed nondeterministic TM that makes full use of the agents' memories! The TM draws its nondeterminism by the nondeterminism inherent in the interaction pattern.

\subsection{Other Previous Work}

Much work concerning the Population Protocol model has been devoted to establishing that the class of computable predicates is precisely the class of \emph{semilinear predicates} \cite{AADFP06,AAER07}. Moreover, in \cite{AADFP06}, the \emph{Probabilistic Population Protocol} model was proposed, in which the scheduler selects randomly and uniformly the next pair to interact. Some recent work has concentrated on performance, supported by this random scheduling assumption (see e.g. \cite{AAE08}). \cite{CDFMS09} proposed a generic definition of probabilistic schedulers and a collection of new fair schedulers, and revealed the need for the protocols to adapt when natural modifications of the mobility pattern occur. \cite{BCCGK09,CS08} considered a huge population hypothesis (population going to infinity), and studied the dynamics, stability and computational power of probabilistic population protocols by exploiting the tools of continuous nonlinear dynamics. Moreover, several extensions of the basic model have been proposed in order to more accurately reflect the requirements of practical systems. In \cite{AACFJP05}, Angluin et al. studied what properties of restricted communication graphs are stably computable, gave protocols for some of them, and proposed an extension of the model with \textit{stabilizing inputs}. In \cite{CMS09-3}, MPP's ability to decide graph properties was studied and it was proven that connectivity is undecidable. Unfortunately, the class of decidable graph languages by MPP remains open. Some other works incorporated agent failures \cite{DFGR06} and gave to some agents slightly increased computational power \cite{BCMRR07} (heterogeneous systems). Recently, Bournez et al. \cite{BCCK08} investigated the possibility of studying population protocols via game-theoretic approaches. For an excellent introduction to the subject of population protocols see \cite{AR07} and for some recent advances mainly concerning mediated population protocols see \cite{CMS09-4}.

\section{Our Results - Roadmap}

In Section \ref{sec:mod}, we begin with a formal definition of the PALOMA model. The section proceeds with a thorough description of the systems' functionality and then provides definitions of \emph{configurations} and \emph{fair executions}. In Section \ref{sec:pred}, \emph{stable computation} of \emph{symmetric predicates on input assignments} is defined. Then the complexity classes $SSPACE(f(n))$, $SNSPACE(f(n))$ (symmetric predicates in $SPACE(f(n))$ and $NSPACE(f(n))$, respectively), and $PLM$ (stably computable predicates by the PALOMA model) are defined, and the section concludes with two examples of stably computable predicates (Subsection \ref{subsec:mult}). Both those predicates are non-semilinear, establishing that the PALOMA model is computationally stronger than the population protocol model. We next study more systematically the computational power of the new model, seeking for an exact characterization for the class $PLM$. We first study the PPALOMA model, a generalization of the PALOMA model that additionally knows the population size $n$ (see Section \ref{sec:pops}). Then, in Section \ref{sec:uids} we focus on another generalization of the PALOMA model, called the IPALOMA model, in which the agents not only know the population size but additionally have already been assigned the unique consecutive ids $\{0,1,\ldots,n-1\}$. We define the corresponding class of stably computable predicates that we call $IPLM$. We begin by studying the relationship between $IPLM$ and $PLM$. Unexpectedly, it turns out that $PLM=IPLM$ (Theorem \ref{the:iplm})! The reason is that the PALOMA model can inform the agents of the population size and assign unique consecutive ids to the agents via an \emph{iterative reinitiation process}. We next prove that $SSPACE(n\log n)$ is a subset of $IPLM$, by showing that in the IPALOMA model the agents can easily organize themselves into a deterministic Turing Machine of $\mathcal{O}(n\log n)$ space.  Thus, the inclusions $SSPACE(n\log n)\subseteq$ $IPLM=PLM$ provide us with a first lower bound for $PLM$. In Section \ref{sec:imlowPLM}, $PLM$'s lower bound is improved to $SNSPACE(n\log n)$. In particular, it is proven that the PALOMA model simulates Guerraoui's and Ruppert's Community Protocol model \cite{GR09}. Unfortunately, the proof of the lower bound on the computational power of the Community Protocol model in \cite{GR09} entirely depends on the simulation of a TM by a Storage Modification Machine, and this dependence is carried to our result. In order to avoid this, in Subsection \ref{subsec:ntm} we go one step further and, by exploiting the techniques of Theorem \ref{the:iplm} (Section \ref{sec:uids}), we show that the PALOMA model can directly simulate a nondeterministic TM of $\mathcal{O}(n\log n)$ space. Moreover, in Section \ref{sec:exact} we show that $SNSPACE(n\log n)$ is an exact characterization for $PLM$, by proving that the corresponding language of any predicate that is stably computable by the PALOMA model can be decided by a nondeterministic TM of $\mathcal{O}(n\log n)$ space. To summarize, the main result of this work is that $PLM$ is equal to $SNSPACE(n\log n)$. Finally, in Section \ref{sec:conc} we conclude and discuss some future research directions.

\section{The Model} \label{sec:mod}

In this section, we formally define the PALOMA model and describe its functionality. In what follows, we denote by $G=(V,E)$ the (directed) communication graph: $V$ is the set of agents, or \emph{population}, and $E$ is the set of permissible ordered pairwise interactions between these agents. We provide definitions for general communication graphs, although in this work we deal with complete communication graphs only. We generally denote by $n$ the population size (i.e. $n\equiv |V|$).

\begin{definition}
A \emph{PALOMA} protocol $\mathcal{A}$ is a 7-tuple $(\Sigma,X,\Gamma,Q,\delta,\gamma,q_0)$ where $\Sigma$, $X$, $\Gamma$ and $Q$ are all finite sets and
\begin{enumerate}
\item $\Sigma$ is the \emph{input alphabet}, where $\#, \sqcup\notin \Sigma$,
\item $X\subseteq\Sigma^{*}$ is the \emph{set of input strings},
\item $\Gamma$ is the \emph{tape alphabet}, where $\#, \sqcup\in \Gamma$ and $\Sigma\subset \Gamma$,
\item $Q$ is the set of \emph{states},
\item $\delta:Q\times\Gamma\rightarrow Q\times\Gamma\times\{L,R\}\times\{0,1\}$ is the \emph{internal transition function},
\item $\gamma:Q\times Q\rightarrow Q\times Q$ is the \emph{external transition function} (or \emph{interaction transition function}), and
\item $q_0\in Q$ is the \emph{initial state}.
\end{enumerate}
\end{definition}

Each agent is equipped with the following:
\begin{itemize}
\item A \emph{sensor} in order to sense its environment and receive a piece of the input (which is an input string from $X$).
\item A \emph{tape} (memory) consisting of $\mathcal{O}(\log n)$ cells. The tape is partitioned into three parts each consisting of $\mathcal{O}(\log n)$ cells: the leftmost part is the \emph{working tape}, the middle part is the \emph{output tape}, and the rightmost part is the \emph{message tape} (we call the parts ``tapes'' because such a partition is equivalent to a 3-tape machine). The last cell of each part contains permanently the symbol $\#$ (we assume that the machine never alters it); it is the symbol used to separate the three tapes and to mark the end of the overall tape.
\item A \emph{control unit} that contains the state of the agent and applies the transition functions.
\item A \emph{head} that reads from and writes to the cells and can move one step at a time, either to the left or to the right.
\item A binary \emph{working flag} either set to $1$ meaning that the agent is \emph{working} internally or to $0$ meaning that the agent is \emph{ready} for interaction.
\end{itemize}

Initially, all agents are in state $q_0$ and all their cells contain the \emph{blank symbol} $\sqcup$ except for the last cell of the working, output, and message tapes that contain the \emph{separator} $\#$. We assume that all agents concurrently receive their sensed input (different agents may sense different data) as a response to a global start signal. The input is a string from $X$ and after reception (or, alternatively, during reception, in an online fashion) it is written symbol by symbol on their working tape beginning from the leftmost cell. During this process the working flag is set to 1 and remains to 1 when this process ends (the agent may set it to 0 in future steps).

When its working flag is set to 1 we can think of the agent working as a usual Turing Machine (but it additionally writes the working flag). In particular, whenever the working flag is set to 1 the internal transition function $\delta$ is applied, the control unit reads the symbol under the head and its own state and updates its state and the symbol under the head, moves the head one step to the left or to the right and sets the working flag to 0 or 1, according to the internal transition function.

We assume that the set of states $Q$ and the tape alphabet $\Gamma$, are both sets whose size is fixed and independent of the population size (i.e. $|Q|=|\Gamma|=\mathcal{O}(1)$), thus, there is, clearly, enough room in the memory of an agent to store both the internal and the external transition functions.

As it is common in the population protocol literature, a \emph{fair adversary scheduler} selects ordered pairs of agents (edges from $E$) to interact. Assume now that two agents $u$ and $\upsilon$ are about to interact with $u$ being the \emph{initiator} of the interaction and $\upsilon$ being the \emph{responder}. Let $f:V\rightarrow \{0,1\}$ be a function returning the current value of each agent's working flag. If at least one of $f(u)$ and $f(\upsilon)$ is equal to $1$, then nothing happens, because at least one agent is still working internally. Otherwise ($f(u)=f(\upsilon)=0$), both agents are ready and an \emph{interaction} is established. In the latter case, the external transition function $\gamma$ is applied, the states of the agents are updated accordingly, the message of the initiator is copied to the message tape of the responder (replacing its contents) and vice versa (the real mechanism would require that each receives the other's message and then copies it to its memory, because instant replacement would make them lose their own message, but this can be easily implemented with $\mathcal{O}(\log n)$ extra cells of memory, so it is not an issue), and finally the working flags of both agents are again set to 1. 

Note that the assumption that the internal transition function $\delta$ is only applied when the working flag is set to 1 is weak. In fact, an equivalent way to model this is to assume that $\delta$ is of the form $\delta:Q\times\Gamma\times\{0,1\}\rightarrow Q\times\Gamma\times\{L,R,S\}\times\{0,1\}$, that it is always applied, and that for all $q\in Q$ and $a\in\Gamma$, $\delta(q,a,0)=(q,a,S,0)$ is satisfied, where $S$ means that the head ``stays put''. The same holds for the assumptions that $\gamma$ is not applied if at least one of the interacting agents is working internally and that the working flags are set to $1$ when some established interaction comes to an end; it is equivalent to an extended $\gamma$ of the form $\gamma:Q^2\times \{0,1\}^2\rightarrow Q^2\times  \{0,1\}^2$, that is applied in every interaction, and for which $\gamma(q_1,q_2,f_1,f_2)=(q_1,q_2,f_1,f_2)$ if $f_1=1$ or $f_2=1$, and $\gamma(q_1,q_2,f_1,f_2)=(\gamma_1(q_1,q_2),\gamma_2(q_1,q_2),1,1)$ if $f_1=f_2=0$, hold for all $q_1,q_2\in Q$, and we could also have further extended $\gamma$ to handle the exchange of messages, but for sake of simplicity we have decided to leave such details out of the model.

Since each agent is a TM (of logarithmic memory), we use the notion of a configuration to capture its ``state''. An \emph{agent configuration} is a quadruple $(q,l,r,f)$, where $q\in Q$, $l,r\in \Gamma^{\mathcal{O}(\log n)}=\{s\in \Gamma^{*}\;|\; |s|=\mathcal{O}(\log n)\}$, and $f\in \{0,1\}$. $q$ is the state of the control unit, $l$ is the string to the left of the head (including the symbol scanned), $r$ is the string to the right of the head, and $f$ is the working flag indicating whether the agent is ready to interact ($f=0$) or carrying out some internal computation ($f=1$). Let $\mathcal{B}$ be the set of all agent configurations. Given two agent configurations $A,A^{\prime}\in \mathcal{B}$, we say that $A$ \emph{yields} $A^{\prime}$ if $A^{\prime}$ follows $A$ by a single application of $\delta$.

A \emph{population configuration} is a mapping $C:V\rightarrow \mathcal{B}$, specifying the agent configuration of each agent in the population. Let $C$, $C^{\prime}$ be population configurations and let $u\in V$. We say that $C$ \emph{yields} $C^{\prime}$ via \emph{agent transition} $u$, denoted $C \stackrel{u}\rightarrow C^{\prime}$, if  $C(u)$ yields $C^{\prime}(u)$ and $C^{\prime}(w)=C(w)$, $\forall w\in V-\{u\}$.

Let $q(A)$ denote the state of an agent configuration $A$, $l(A)$ its string to the left of the head including the symbol under the head, $r(A)$ its string to the right of the head, and $f(A)$ its working flag. Given two population configurations $C$ and $C^{\prime}$, we say that $C$ \emph{yields} $C^{\prime}$ via \emph{encounter} $e=(u,\upsilon)\in E$, denoted $C\stackrel{e}\rightarrow C^{\prime}$, if one of the following two cases holds:\\

\noindent Case 1:
\begin{itemize}
\item $f(C(u))=f(C(\upsilon))=0$ which guarantees that both agents $u$ and $\upsilon$ are ready for interaction under the population configuration $C$.
\item $r(C(u))$ and $r(C(\upsilon))$ are precisely the message strings of $u$ and $\upsilon$, respectively (this is a simplifying assumption stating that when an agent is ready to interact its head is over the last $\#$ symbol, just before the message tape),
\item $C^{\prime}(u) = (\gamma_1(q(C(u)),q(C(\upsilon))),l(C(u)), r(C(\upsilon)), 1)$,
\item $C^{\prime}(\upsilon) = (\gamma_2(q(C(u)),q(C(\upsilon))),l(C(\upsilon)), r(C(u)), 1)$, and
\item $C^{\prime}(w)=C(w)$, $\forall w\in V-\{u,\upsilon\}$.
\end{itemize}

\noindent Case 2:
\begin{itemize}
\item $f(C(u))=1$ or $f(C(\upsilon))=1$, which means that at least one agent between $u$ and $\upsilon$ is working internally under the population configuration $C$, and
\item $C^{\prime}(w) = C(w)$, $\forall w\in V$. In this case no effective interaction takes place, thus the population configuration remains the same.
\end{itemize}

Generally, we say that $C$ \emph{yields} (or \emph{can go in one step to}) $C^{\prime}$, and write $C\rightarrow C^{\prime}$, if $C\stackrel{e}\rightarrow C^{\prime}$ for some $e\in E$ (via encounter) or $C\stackrel{u}\rightarrow C^{\prime}$ for some $u\in V$ (via agent transition), or both. We say that $C^{\prime}$ is \emph{reachable} from $C$, and write $C\stackrel{*}\rightarrow C^{\prime}$ if there is a sequence of population configurations $C=C_0,C_1,\ldots,C_t=C^{\prime}$ such that $C_i\rightarrow C_{i+1}$ holds for all $i\in \{0,1,\ldots,t-1\}$. An \emph{execution} is a finite or infinite sequence of population configurations $C_0,C_1\dots$, so that $C_i\rightarrow C_{i+1}$. An infinite execution is \emph{fair} if for all population configurations $C$, $C^{\prime}$ such that $C\rightarrow C^{\prime}$, if $C$ appears infinitely often then so does $C^{\prime}$. A \emph{computation} is an infinite fair execution.

Note that the PALOMA model \emph{preserves uniformity}, because $X$, $\Gamma$ and $Q$ are all finite sets whose cardinality is independent of the population size. Thus, protocol descriptions have also no dependence on the population size. Moreover, PALOMA protocols are \emph{anonymous}, since initially all agents are identical and have no unique identifiers.

%\begin{definition}
%The \emph{basic} PALOMA model is the PALOMA model restricted to complete communication graphs.
%\end{definition}
%In what follows, and if not otherwise stated, by PALOMA we mean basic PALOMA, that is, we generally %consider complete communication graphs (although, sometimes ``basic'' appears before ``PALOMA'' as a %reminder).

\section{Stably Computable Predicates} \label{sec:pred}

The predicates that we consider here are of the following form. The input to each agent is simply some string $s\in X$. An \emph{input assignment} $x$ is a mapping from $V$ to $X$ assigning an input string to each agent of the population. Let $\mathcal{X}=X^{V}=X^{n}$ be the set of all input assignments. Any mapping $p: \mathcal{X}\rightarrow \{0,1\}$ is a \emph{predicate on input assignments}. A predicate on input assignments $p$ is called \emph{symmetric} if for every $x=(s_1,s_2,\ldots,s_n)\in\mathcal{X}$ and any $x^{\prime}$ which is a permutation of $x$'s components, it holds that $p(x)=p(x^{\prime})$ (in words, permuting the input strings does not affect the predicate's outcome). A population configuration $C$ is called \emph{output stable} if for every configuration $C^{\prime}$ that is reachable from $C$ it holds that $O(C^{\prime}(u))=O(C(u))$ for all $u\in V$, where $O(C(u))$ denotes the contents of the output tape of agent $u$ under configuration $C$ (i.e. a string). In simple words, no agent changes its output in any subsequent step and no matter how the computation proceeds. A predicate on input assignments $p$ is said to be \emph{stably computable} by a PALOMA protocol $\mathcal{A}$ if, for any input assignment $x$, any computation of $\mathcal{A}$ contains an output stable configuration in which all agents have $p(x)$ written on their output tape.

We say that a predicate on input assignments $p$ belongs to $SPACE(f(n))$ ($NSPACE(f(n))$) if there exists some deterministic (nondeterministic, resp.) TM that gets $x\in \mathcal{X}$ as input (e.g. as a $n$-vector of input strings), if $p(x)=1$ accepts and if $p(x)=0$ rejects by using $\mathcal{O}(f(n))$ space.

\begin{remark}
All agents are identical and do not initially have unique ids, thus, stably computable predicates by the PALOMA model on complete communication graphs have to be symmetric.
\end{remark}

\begin{definition}
Let $SSPACE(f(n))$ and $SNSPACE(f(n))$ be $SPACE(f(n))$'s and $NSPACE($ $f(n))$'s restrictions to symmetric predicates, respectively.
\end{definition}

\begin{definition}
Let $PLM$ denote the class of all symmetric predicates that are stably computable by the PALOMA model.
\end{definition}
Our main result in this paper is the following exact characterization for $PLM$: $PLM$ is equal to $SNSPACE(n\log n)$.

\subsection{Two Examples} \label{subsec:mult}

We show here that there exist PALOMA protocols that stably compute the predicates $(N_c=N_a\cdot N_{b})$ ($N_q$ denotes the number of agents with input $q$) and $(N_1=2^{t})$, where $t\in \bbbz_{\geq 0}$.

\subsubsection{Multiplication of Variables} \label{subsubsec:mult1}

\noindent \\ \\ We begin by presenting a PALOMA protocol that stably computes the predicate $(N_c=N_a\cdot N_{b})$ (on the complete communication graph of $n$ nodes) that is, all agents eventually decide whether the number of $c$s in the input assignment is the product of the number of $a$s and the number of $b$s. We give a high-level description of the protocol.

Initially, all agents have one of $a$, $b$ and $c$ written on the first cell of their working memory (according to their sensed value). That is, the set of input strings is $X=\Sigma=\{a,b,c\}$. Each agent that receives input $a$ goes to state $a$ and becomes ready for interaction (sets its working flag to 0). Agents in state $a$ and $b$ both do nothing when interacting with agents in state $a$ and agents in state $b$. An agent in $c$ initially creates in its working memory three binary counters, the $a$-counter that counts the number of $a$s, the $b$-counter, and the $c$-counter, initializes the $a$ and $b$ counters to 0, the $c$-counter to 1, and becomes ready. When an agent in state $a$ interacts with an agent in state $c$, $a$ becomes $\bar{a}$ to indicate that the agent is now sleeping, and $c$ does the following (in fact, we assume that $c$ goes to a special state $c_{a}$ in which it knows that it has seen an $a$, and that all the following are done internally, after the interaction; finally the agent restores its state to $c$ and becomes again ready for interaction): it increases its $a$-counter by one (in binary), multiplies its $a$ and $b$ counters, which can be done in binary in logarithmic space (binary multiplication is in $LOGSPACE$), compares the result with the $c$-counter, copies the result of the comparison to its output tape, that is, 1 if they are equal and 0 otherwise, and finally it copies the comparison result and its three counters to the message tape and becomes ready for interaction. Similar things happen when a $b$ meets a $c$ (interchange the roles of $a$ and $b$ in the above discussion). When a $c$ meets a $c$, the responder becomes $\bar{c}$ and copies to its output tape the output bit contained in the initiator's message. The initiator remains to $c$, adds the $a$-counter contained in the responder's message to its $a$-counter, the $b$ and $c$ counters of the message to its $b$ and $c$ counters, respectively, multiplies again the updated $a$ and $b$ counters, compares the result to its updated $c$ counter, stores the comparison result to its output and message tapes, copies its counters to its message tape and becomes ready again. When a $\bar{a}$, $\bar{b}$ or $\bar{c}$ meets a $c$ they only copy to their output tape the output bit contained in $c$'s message and become ready again (eg $\bar{a}$ remains $\bar{a}$), while $c$ does nothing.

Note that the number of $c$s is at most $n$ which means that the $c$-counter will become at most $\lceil \log n\rceil$ bits long, and the same holds for the $a$ and $b$ counters, so there is enough room in the tape of an agent to store them.

\begin{theorem}
The above PALOMA protocol stably computes the predicate $(N_c=N_a\cdot N_{b})$.
\end{theorem}
\begin{proof}
Given a fair execution, eventually only one agent in state $c$ will remain, its $a$-counter will contain the total number of $a$s, its $b$-counter the total number of $b$s, and its $c$-counter the total number of $c$s. By executing the multiplication of the $a$ and $b$ counters and comparing the result to its $c$-counter it will correctly determine whether $(N_c=N_a\cdot N_{b})$ holds and it will store the correct result (0 or 1) to its output and message tapes. At that point all other agents will be in one of the states $\bar{a}$, $\bar{b}$, and $\bar{c}$. All these, again due to fairness, will eventually meet the unique agent in state $c$ and copy its correct output bit (which they will find in the message they get from $c$) to their output tapes. Thus, eventually all agents will output the correct value of the predicate.
\qed
\end{proof}

\begin{corollary}
The PALOMA model is strictly stronger than the population protocol model.
\end{corollary}
\begin{proof}
PALOMA simulates PP and stably computes $(N_c=N_a\cdot N_{b})$ which is non-semilinear.
\qed
\end{proof}

\subsubsection{Power of 2} \label{subsubsec:pow}

\noindent \\ \\ We now present a PALOMA protocol that stably computes the predicate $(N_1=2^{t})$, where $t\in \bbbz_{\geq 0}$, on the complete communication graph of $n$ nodes, that is, all agents eventually decide whether the number of $1$s in the input assignment is a power of 2.

The set of input strings is $X=\Sigma=\{0,1\}$. All agents that receive 1 create a binary $1$-counter to their working tape and initialize it to $1$. Moreover, they create a binary $next\_pow\_of2$ block and set it to $2$. Finally, they write 1 (which is interpreted as ``true'') to their output tape, and copy the $1$-counter and the output bit to their message tape before going to state 1 and becoming ready. Agents that receive input 0 write 0 (which is interpreted as ``false'') to their output tape, go to state 0, and become ready. Agents in state 0 do nothing when interacting with each other. When an agent in state 0 interacts with an agent in state 1, then 0 simply copies the output bit of 1. When two agents in state 1 interact, then the following happen: the initiator sets its $1$-counter to the sum of the responder's $1$-counter and its own $1$-counter and compares its updated value to $next\_pow\_of2$. If it is less than $next\_pow\_of2$ then it writes 0 to the output tape. If it is equal to $next\_pow\_of2$ it sets $next\_pow\_of2$ to $2\cdot next\_pow\_of2$ and sets its output bit (in the output tape) to 1. If it is greater than $next\_pow\_of2$, then it starts doubling $next\_pow\_of2$ until $1$-counter $\geq next\_pow\_of2$ is satisfied. If it is satisfied by equality, then it doubles $next\_pow\_of2$ one more time and writes 1 to the output tape. Otherwise, it simply writes 0 to the output tape. Another implementation would be to additionally send $next\_pow\_of2$ blocks via messages and make the initiator set $next\_pow\_of2$ to the maximum of its own and the responder's $next\_pow\_of2$ blocks. In this case at most one doubling would be required. Finally, in both implementations, the initiator copies the output bit and the $1$-counter to its message tape (in the second implementation it would also copy $next\_pow\_of2$ to the message tape), remains in state $1$, and becomes ready. The responder simply goes to state $\bar{1}$ and becomes ready. An agent in state $\bar{1}$ does nothing when interacting with an agent in state 0 and vice versa. When an agent in state $\bar{1}$ interacts with an agent in state 1, then $\bar{1}$ simply copies the output bit of 1.

Note that $next\_pow\_of2$ can become at most 2 times the number of $1$s in the input assignment, and the latter is at most $n$. Thus, it requires at most $\lceil \log 2n\rceil$ bits of memory. Either way, we can delay the multiplication until another 1 appears, in which case we need at most $\lceil \log n\rceil$ bits of memory for storing $next\_pow\_of2$ (the last unnecessary multiplication will never be done).

\subsection{Computation Under Other Input Conventions} \label{subsec:conv}

The \emph{multiplicative integer input convention} assumes that upon initialization, each agent receives a value and a mark. An integer $N_x$ is retrieved by calculating the product $x_1\cdot x_2\dots x_n$ of the $x_i$ values stored at different agents marked as $x$. Under this assumption, there is a simple way to compute the predicate $N_c=N_a\cdot N_b$. Because the input convention matches the predicate operator, it is equivalent to compute the predicate $N_a=N_b$ instead.

Upon initialization each agents is marked as either $a$ or $b$ and is assigned a value ($value_a$ or $value_b$). Moreover, each  agent maintains a $flag$ variable indicating the correctness of the predicate. Initially, $flag=\bot$ for each agent.

Interactions between agents have the following effects:

\begin{description}
	\item[$\{a,b\}$] Then $value_a=\frac{value_a}{gcd(value_a,value_b)}$ and  $value_b=\frac{value_b}{gcd(value_a,value_b)}$. If after this update $value_a=value_b=1$ then $flag_a=flag_b=\top$. Otherwise, $flag_a=flag_b=\bot$
	\item[$\{a,a\}, \{b,b\}$] If either $flag$ is set to $\bot$ then both flags are set to $\bot$.
\end{description}

\begin{lemma}
$N_a=N_b$ iff for each configuration $C$ there exist $a_i$, $b_j$ such that $gcd(value_{a_i}, value_{b_j})\neq 1$ or $a_i=b_j=1 \forall i,j$
\end{lemma}
\begin{proof}
\begin{description}
	\item[$\rightarrow$] Assume on the contrary that there is no such pair and that there exists a $b_i\neq 0$. Then, $\forall j$, $a_j$, $b_i$ are coprimes. There is a theorem stating that if $a$, $b_1$ are coprimes and $a,b_2$ are coprimes, then $a,b_1\cdot b_2$ are coprimes. Applying this theorem inductively we can prove that $N_a$ and $b_i$ are coprimes. Then $N_b$ is a multiple of $b_i$, while $N_a$ is not, thus $N_a$ and $N_b$ can't be equal.
	\item[$\leftarrow$] If there exist $a_i$, $b_j$ such that $gcd(value_{a_i}, value_{b_j})\neq 1$ or $a_i=b_j=1 \forall i,j$, then we have successfully divided $N_a$, $N_b$ with the same divisor $X=\prod_{\forall i,j}{gcd(value_{a_i}, value_{b_j})}$ until $a_i=b_j=1 \forall i,j$. The new integers represented by the  diffused $1$s are $N'_a=N_b=1$, and since $N'_x=\frac{N_x}{X}$, we conclude that $N_a=N_b$.
\end{description}
\qed
\end{proof}

\begin{lemma}
$N_a=N_b$ iff eventually for each agent $i$, $value_i=1$
\end{lemma}
\begin{proof}
	\item[$\rightarrow$] If $N_a=N_b$, then by the previous lemma $gcd(value_{a_i}, value_{b_j})\neq 1$ or $a_i=b_j=1 \forall i,j$. In the second case there is nothing to prove. In the first case, observe that each time $gcd(value_{a_i}, value_{b_j})\neq 1$ a division takes place that will cut its operands at least in half. In discrete context, this sequence of divisions must end when $a_i=b_j=1 \forall i,j$.
	\item[$\leftarrow$] If for each agent $i$, $value_i=1$ then we have successfully divided $N_a$, $N_b$ with the same divisor $X=\prod_{\forall i,j}{gcd(value_{a_i}, value_{b_j})}$ until $a_i=b_j=1 \forall i,j$. The new integers represented by the diffused $1$s are $N'_a=N_b=1$, and since $N'_x=\frac{N_x}{X}$, we conclude that $N_a=N_b$.
\qed
\end{proof}
\begin{lemma}
If $N_a=N_b$, eventually each agent $i$ will set $flag_i$ to $\top$
\end{lemma}
\begin{proof}
If $N_a=N_b$ then by the previous lemma each agent $i$ will have $value_i=1$. In that case, the $gcd$ between two values will always be $1$, thus interactions between $a$s and $b$s will only generate $\top$ values. Due to fairness, during the execution a sequence of interactions will occur in which any $\{a,a\}$ or $\{b,b\}$ interaction that would  generate $\bot$ values can be delayed so that the participant $i$ will first interact with an agent $-i$. This interaction will force $value_i$ to $\top$ thus when the $\{i,i\}$ interaction eventually takes place no $\bot$ values are present so no $\bot$ values are generated.
\qed
\end{proof}
\begin{lemma}
If $N_a\neq N_b$, eventually each agent $i$ will set $flag_i$ to $\bot$
\end{lemma}
\begin{proof}
By a previous lemma we know that if $N_a\neq N_b$ then there exists an agent $b_i$ that its value will get permanent and different from $1$. Due to fairness, at some point all $a_i$ agents will consecutively interact with $b_i$ and set their flags to $\bot$. Then it is easy to observe that in any interaction that will follow, no $\top$ value will be assigned to any $flag$.
\qed
\end{proof}
This protocol is quite inefficient as convergence is guaranteed only when a specific sequence of $n$ interactions takes place. We can probably optimize it but that would make it more complex. Finally, notice the absence of memory assumptions. Indeed, the interactions are so straightforward that this protocol is applicable even to the PP model.

\section{Knowing the Population Size} \label{sec:pops}

We prove now that if the agents know the population size $n$ then any symmetric predicate in $SPACE(n\log n)$ is stably computable by the PALOMA model when the communication graph is directed and complete.

\begin{lemma} \label{lem:ids}
If a PALOMA protocol knows the population size $n$ then it can assign to the agents the unique ids $0,1,\ldots,n-1$ and, moreover, all agents will eventually know that the process has terminated.
\end{lemma}
\begin{proof}
Note that $n$ can be stored with $\lceil \log n\rceil$ bits, thus we assume w.l.o.g. that it is already stored in some block of the agents' working memories (e.g. we can assume that $\Sigma=\{0,1\}$ and $X=\{b(n)\}$, where $b(n)$ denotes the binary string representation of $n$, thus all agents get $b(n)$ as their input). All agents initialize an id to 0 in their working memory, copy it to their message memory, remain to state $q_0$, and become ready. When two agents interact, the external transition function alters their states in order to inform the initiator that it was the initiator and the responder that it was the responder (e.g. the rule $(q_0,q_0)\rightarrow (q_1,q_2)$ suffices). After the interaction, the initiator compares its id to the received id (that of the responder), and if they are equal it increases its own id by one. If its updated id is less than $n-1$, it copies it in the message tape, goes to $q_0$ again and becomes ready. If it is equal to $n-1$ it goes to $q_f$ and becomes ready. $q_f$ is then eventually propagated to the whole population because $\gamma(q_f,\cdot)=(q_f,q_f)$ and $\gamma(\cdot,q_f)=(q_f,q_f)$ holds, and all agents are informed that the process has been competed successfully. The responder simply restores its message tape to contain again only its own id (if we assume that its own message is not lost then this step is not necessary, but both ways are equivalent) goes again to $q_0$ and becomes ready.

We claim that the above process is correct. Each agent's id can only be incremented. Moreover, an id that has appeared in the population can never be completely eliminated (it is only incremented when sure that it also exists in another agent). As long as id $n-1$ has not appeared, by the pigeonhole principle, there will be at least two agents with the same id. Thus, eventually (in a finite number of steps), two such agents will interact and one of them will increase its counter by one. Clearly, the above process must end in a finite number of steps with an agent having id $n-1$. When this happens, the agents are assigned the unique ids $0,1,\ldots,n-1$. If not, then at least one id $i < n-1$ is missing from the population. But $i$ should have appeared because then $n-1$ could not have been reached. But this is a contradiction, because once an id appears then it can never be completely eliminated.
\qed
\end{proof}

\begin{definition}
Let \emph{PPALOMA} (`P' standing for ``Population size'') be the extension of PALOMA in which each agent knows the population size (it can read it from a read-only block of its memory).
\end{definition}

\begin{definition}
Let $PPLM$ denote the class of all symmetric predicates that are stably computable by some PPALOMA protocol.
\end{definition}

\begin{theorem} \label{the:psize}
$SSPACE(n\log n)$ is a subset of $PPLM$.
\end{theorem}
\begin{proof}
Let $p:\mathcal{X}\rightarrow \{0,1\}$ be any predicate in $SSPACE(n\log n)$ and $\mathcal{M}$ be the deterministic TM that decides $p$ by using $\mathcal{O}(n\log n)$ space. We construct a PPALOMA protocol $\mathcal{A}$ that stably computes $p$ by exploiting its knowledge of the population size. Let $x$ be any input assignment in $\mathcal{X}$. Each agent receives its input string according to $x$ (e.g. $u$ receives string $x(u)$). Now the agents obtain unique ids according to the protocol presented in the proof of Lemma \ref{lem:ids}. The agent that has obtained the unique id $n-1$ starts simulating $\mathcal{M}$.

If at some point the transition function of $\mathcal{M}$ moves the head to the right, but the agent's working memory has no other symbol to read (e.g. it reads $\sqcup$), then it writes $n-2$ and $\mathcal{M}$'s current state to its message tape and becomes ready. When the unique agent with id $n-2$ interacts with an agent that has $n-2$ in its message tape, it starts simulating $\mathcal{M}$ by putting the head over the first symbol in its working tape and assuming that the state of $\mathcal{M}$ is the state that it found written on the message it received. Generally, whenever an agent with id $i>0$ cannot continue simulating $\mathcal{M}$ to the right it passes control to the agent with id $i-1$. Additionally, if an agent with id $j < n-1$ that simulates $\mathcal{M}$ ever reaches its leftmost cell and the transition function of $\mathcal{M}$ wants to move the head left, then it informs agent with id $j+1$ to continue the simulation (when they, eventually, interact).

Note now, that at some point $\mathcal{M}$ may want to use its (initially) blank cells to write symbols. This is handled by $\mathcal{A}$ in the following manner. $\mathcal{A}$ first starts using the blank cells of the agent with id 0. If more are needed, it goes to agent $n-1$, writes a separator to the first blank cell (naturally we can assume that the separator is already there from $\mathcal{A}$'s initialization step, because also $\mathcal{M}$ uses separators, like `,'s, between the different inputs of the symmetric predicate) and starts using those blank cells. Then it can use those of $n-2,n-1,\ldots$ and so on until the last blank cell of agent with id $1$. In fact, in this approach the distributed memory of $\mathcal{M}$ can be read by beginning from agent $n-1$ and reading all blocks (the parts of the working memories that are used for the simulation) before the separators until agent $1$. Then we read the whole simulation block of agent 0, proceed with the simulation block after the separator of agent $n-1$ and continue with those blocks (after the separator) of all agents until agent with id 1. Thus, the memory is read in a cyclic fashion.

Another approach would be, to initially transfer (shift) the concatenation of all agents inputs (separated by some symbol, e.g. `,') to agents $0,1,2,\ldots$. Now, the first $k$ agents will contain input data, agent with id $k-1$ will probably also have some blank cells and the remaining agents will contain only blank cells. In this case, simulation starts from agent $0$ and $\mathcal{M}$'s tape can be read sequentially from agent $0$ to agent $n-1$.

Whenever, during the simulation, $\mathcal{M}$ accepts, then $\mathcal{A}$ also accepts; that is, the agent that detects $\mathcal{M}$'s acceptance, writes 1 to its output tape and informs all agents to accept. If $\mathcal{M}$ rejects, it also rejects. Finally, note that $\mathcal{A}$ simulates $\mathcal{M}$ not necessarily on input $x=(s_1,s_2,\ldots,s_n)$ but on some $x^{\prime}$ which is a permutation of $x$. The crucial remark that completes the proof is that $\mathcal{M}$ accepts $x$ if and only if it accepts $x^{\prime}$, because $p$ is symmetric.
\qed
%Let $s$ be any string in $L$. Obviously, w.l.o.g. $|s|=\mathcal{O}(n\log n)$ because this is the greatest number of %symbols that can be stored in only $\mathcal{O}(n\log n)$ cells. We partition $s$ into $n$ substrings %$s_{0},\ldots,s_{n-1}$, each consisting of $\mathcal{O}(\log n)$ symbols, where $s=s_{0}s_{1}\ldots s_{n-1}$. %Moreover, we add $i$ in binary to the front of each $s_i$ and use a special symbol (e.g. $\$$) as a separator %between $i$ and $s_i$, that is we create $s^{\prime}_{i}=b(i)\$s_i$ for all $i$ ($b(i)$ simply denotes the binary %string representation of $i$).
\end{proof}

\section{Assigning Unique IDs by Reinitiating Computation} \label{sec:uids}

In this section, we first prove that PALOMA protocols can assume the existence of unique consecutive ids and knowledge of the population size. In particular, in Theorem \ref{the:iplm} we prove that any PALOMA protocol that assumes the existence of unique consecutive ids and knows the population size, can be composed with a PALOMA protocol that correctly assigns these unique ids to the agents and informs them of the correct population size, without assuming any initial knowledge of none of them. We then exploit this result to obtain a first lower bound for $PLM$.

\begin{definition}
Let \emph{IPALOMA} (`I' standing for ``Ids'') be the extension of PALOMA in which the agents have additionally the read-only unique ids $\{0,1,\ldots, n-1\}$ and in which each agent knows the population size (it can read it from a read-only block of its memory).
\end{definition}

\begin{definition}
Let $IPLM$ denote the class of all symmetric predicates that are stably computable by the IPALOMA model.
\end{definition}

\begin{theorem} \label{the:iplm}
$PLM=IPLM$.
%Let $\mathcal{B}$ be a basic PALOMA protocol that stably computes the predicate $p:\mathcal{X}\rightarrow \{0,1\}$.
\end{theorem}
\emph{Intuitive Proof Argument}. To simplify reading the formal proof that follows, we first give an intuitive description. The hard part is $IPLM\subseteq PLM$. The idea is to show that there exists a PALOMA protocol $\mathcal{I}$ that assigns unique consecutive ids and informs all agents of the correct population size, although initially all agents have the same id and know the same, wrong, population size, which are both equal to zero. Then any IPALOMA protocol $\mathcal{A}$, running as a subroutine, can use those data provided by $\mathcal{I}$ in order to be correctly executed.

$\mathcal{I}$ operates as follows. Whenever two agents with the same id interact, one of them increases its id by one and the other keeps its id. Moreover, they now both think that the correct population size is equal to the updated id plus one. By using some arguments similar to those in Lemma \ref{lem:ids}, it is not so hard to show that this process correctly assigns the unique ids $\{0,1,\ldots,n-1\}$ (in a finite number of steps, and due to fairness). In fact, it always first assigns the ids $\{0,1,\ldots,n-3,n-2,n-2\}$ and the last interaction that modifies some id will always be between some agents $u$ and $\upsilon$ that both have the id $n-2$. These now obtain the ids $n-1$ and $n-2$, respectively, and are the only agents that know the correct population size.

Of course, $\mathcal{I}$ cannot detect termination of the id-assignment process, because, otherwise, it is easy to show that it would wrongly detect termination in some component with $t<n$ agents that does not interact with the rest population for a long period of time (it would be possible that the ids $\{0,1,\ldots,t-1\}$ would have been assigned in this component and then those agents would have committed an erroneous termination). To overcome this fact, we do not wait for some termination criterion to be satisfied before executing $\mathcal{A}$. $\mathcal{I}$ allows $\mathcal{A}$ to be executed even while the ids are still totaly wrong. The trick is that, whenever some id is incremented, $\mathcal{I}$ knows that $\mathcal{A}$ was running up to this point without unique ids, and with a wrong population size. By exploiting this knowledge, we make $\mathcal{I}$ at this point reinitialize all agents (restores the input and erases all data produced by $\mathcal{A}$, without altering the ids) and inform them of the new population size with the hope that now the id-assignment process has come to an end. Note that the reinitialization process will be executed $n(n-1)/2$ times, until the ids stabilize, because this is the number of times that some id is incremented. But when the last id-modification takes place, all agents have the correct unique consecutive ids and two agents $u$, $\upsilon$ know the correct population size. Thus, if $\mathcal{I}$ does the above correctly, the population size will be now propagated to all agents and all of them will start executing $\mathcal{A}$ from the beginning. Since $\mathcal{A}$ now reads the correct ids and the correct population size, it will run correctly, like if the id-assignment process had never taken place, and although it was several times running in a wrong manner. We will also have to guarantee that $\mathcal{I}$ does not allow reinitialized agents to have some effective interaction (subroutine $\mathcal{A}$ gets executed) with agents that have not yet been reinitialized, because, otherwise, possibly outdated data of $\mathcal{A}$ would get mixed with those of the restored correct execution. Finally, we always execute $\mathcal{A}$ for a constant number of steps, because in previous wrong rounds it is possible that conflict of inconsistent data takes place (e.g. two agents have consecutive ids but have obtained them in different components), and this could make some machine fall into some infinite loop, and, thus, stop interacting (always busy).
\qed

\begin{proof}
$PLM\subseteq IPLM$ holds trivially (IPALOMA is PALOMA with extra capabilities), so it suffices to show that $IPLM\subseteq PLM$. The rest of the proof is structured in four technical lemmata, which follow, each with its own proof. We do this for sake of clarity and readability. Pick any $p\in IPLM$. Let $\mathcal{A}$ be the IPALOMA protocol that stably computes it. We present a PALOMA protocol $\mathcal{B}$ that stably computes $p$. $\mathcal{B}$ consists of a procedure $\mathcal{I}$ containing protocol $\mathcal{A}$ as a subroutine (see Protocol \ref{prot:ids}). $\mathcal{I}$ is always executed and its job is to assign unique ids to the agents, to inform them of the correct population size and to control $\mathcal{A}$'s execution (e.g. restarts its execution if needed). $\mathcal{A}$, when $\mathcal{I}$ allows its execution and for as many steps as it allows, simply reads the unique ids and the population size provided by $\mathcal{I}$ and executes itself normally. We first describe $\mathcal{I}$'s functionality and then prove that it eventually correctly assigns unique ids and correctly informs the agents of the population size, and that when this process comes to a successful end, it restarts $\mathcal{A}$'s execution in all agents without allowing non-reinitialized agents to communicate with the reinitialized ones. Thus, at some point, $\mathcal{A}$ will begin its execution reading the correct unique ids and the correct population size (provided by $\mathcal{I}$), thus, it will get correctly executed and will stably compute $p$.

We begin by describing $\mathcal{I}$'s variables. $id$ is the variable storing the id of the agent (from which $\mathcal{A}$ reads the agents' ids), $sid$ the variable storing the $id$ that an agent writes in the message tape in order to send it, and $rid$ the variable storing the $id$ that an agent receives via interaction. In order to simplify code we make a convention. We assume that all variables used for sending information, like $sid$, preserve their value in future interactions unless altered by the agent. Initially, $id=sid=0$ for all agents. All agents have an input backup variable $binput$ which they initially set to their input string and make it read-only. Thus, each agent has always available its input via $binput$ even if the computation has proceeded. $working$ represents the block of the working tape that $\mathcal{A}$ uses for its computation and $output$ represents the contents of the output tape. $initiator$ is a binary flag that after every interaction becomes true if the agent was the initiator of the interaction and false otherwise (this is easily implemented by exploiting the external transition function). $ps$ is the variable storing the population size, $sps$ the one used to put it in a message, and $rps$ the received one. Initially, $ps=sps=0$. $search\_for$ contains the id that the agent is searching for to interact with, and $ssearch\_for$, $rsearch\_for$ are defined similarly as before. $search\_for$ and $ssearch\_for$ are initially equal to $-1$ which means \emph{not searching}; the value $-2$ is interpreted as \emph{waiting}.

\algsetup{indent=2em}
\floatname{algorithm}{Protocol}
\renewcommand{\algorithmiccomment}[1]{// #1}
\begin{algorithm}[H]
  \caption{$\mathcal{I}$}\label{prot:ids}
  \begin{algorithmic}[1]
    %\REQUIRE The set of initial configurations $C_{I}$ and the transition relation $\Delta$.
    %\ENSURE The transition graph $G_{r}$.
    \medskip
    \IF [two agents with the same ids interact]{$rid==id$}
       \IF [the initiator]{$initiator==1$}
          \STATE $id\leftarrow id+1$, $sid\leftarrow id$ \COMMENT {increases its id by one and writes it in the message tape}
	  \STATE $search\_for\leftarrow 0$, $ssearch\_for\leftarrow 0$ \COMMENT {will start searching for id 0 to reset it}
	  \STATE $ps\leftarrow id+1$, $sps\leftarrow ps$ \COMMENT {sets the population size equal to its updated id plus 1}
       \ELSE [the responder]
	  \STATE $search\_for\leftarrow -2$, $ssearch\_for\leftarrow -2$ \COMMENT {starts waiting}
	  \STATE $ps\leftarrow id+2$, $sps\leftarrow ps$
       \ENDIF
       \STATE \COMMENT {both clear their working block and copy their input string into it}
       \STATE \COMMENT {they also clear their output tape}
       \STATE $working\leftarrow binput$, $output\leftarrow \emptyset$
    \STATE \COMMENT {two agents whose ids differ by one interact}
    \ELSIF {$rid==id-1$ \OR ($id==0$ \AND $rid\geq ps-1$)}
      \STATE \COMMENT {the one with the greatest id}
      \IF [the other was searching for it]{$rsearch\_for==id$}
         \STATE $working\leftarrow binput$, $output\leftarrow \emptyset$
	 \IF {$search\_for == -2$ \AND $rps==ps$}
	    \STATE $search\_for\leftarrow-1$, $ssearch\_for\leftarrow-1$ \COMMENT {stops waiting}
	    \PRINT $\mathcal{A}$ for $c=\mathcal{O}(1)$ steps ignoring the received data
	 \ELSIF {$rps>ps$}
	    \STATE $search\_for\leftarrow id+1$, $ssearch\_for\leftarrow id+1$
	    \STATE $ps\leftarrow rps$, $sps\leftarrow ps$
	 \ENDIF
      \ELSIF {$search\_for == -1$ \AND $rsearch\_for==-1$}
	 \PRINT $\mathcal{A}$ for $c=\mathcal{O}(1)$ steps
      \ENDIF
    \ELSIF {$rid==id+1$  \OR ($rid==0$ \AND $id==ps-1$)}
       \STATE \COMMENT {the one with the smallest id}
       \IF [found what it was searching for]{$search\_for==rid$}
          \STATE $search\_for\leftarrow -1$, $ssearch\_for\leftarrow -1$ \COMMENT {stops searching}
	  \PRINT $\mathcal{A}$ for $c=\mathcal{O}(1)$ steps ignoring the received data
       \ELSIF {$search\_for == -1$ \AND $rsearch\_for==-1$}
	  \PRINT $\mathcal{A}$ for $c=\mathcal{O}(1)$ steps
       \ENDIF
    \ELSE [the interaction is not part of the virtual ring]
       \IF {$search\_for == -1$ \AND $rsearch\_for==-1$ \AND $ps==rps$}
         \PRINT $\mathcal{A}$ for $c=\mathcal{O}(1)$ steps
       \ENDIF
    \ENDIF
  \end{algorithmic}
\end{algorithm}

\begin{lemma} \label{lem:ids1}
(i) $\mathcal{I}$ assigns the ids $\{0,1,\ldots,n-1\}$ in a finite number of steps. (ii) The id-assignment process ends with an interaction $(u,\upsilon)$ of two agents $u$ and $\upsilon$ that both have the id $n-2$. (iii) This is the last interaction that modifies some agent's id. (iv) When this interaction happens, $u$ and $\upsilon$ know $n$ and all other agents know a population size that is strictly smaller than $n$.
\end{lemma}
\begin{proof}
(i) Initially, all agents have the id $0$. Each agent's id can only be incremented. Moreover, an id that has appeared in the population can never be completely eliminated (it is only incremented when sure that it also exists in another agent). As long as id $n-1$ has not appeared, by the pigeonhole principle, there will be at least two agents with the same id. Thus, eventually (in a finite number of steps), due to fairness, two such agents will interact and one of them will increase its id by one. Clearly, the above process must end in a finite number of steps with an agent having id $n-1$. When this happens, the agents are assigned the unique ids $\{0,1,\ldots,n-1\}$. If not, then at least one id $i < n-1$ is missing from the population. But $i$ should have appeared because then $n-1$ could not have been reached. But this is a contradiction, because once an id appears then it can never be completely eliminated.

(ii) Assume not. Then it must end by some interaction between two agents $u$ and $\upsilon$ that both have the same id $i<n-2$. After the interaction, $u$ has the id $i+1$, $\upsilon$ the id $i$, and the agents in general have the uids $\{0,1,\ldots,n-1\}$. This implies that id $i+1$ did not exist in the population just before this interaction. But for $n-1$ to exist it must be the case that $i+1 < n-1$ had appeared at some point. But then it could have never been completely eliminated, which is a contradiction.

(iii) Just after the unique consecutive ids $\{0,1,\ldots,n-1\}$ have been assigned, no agents have the same id. Ids are only modified when two agents with the same id interact. Thus, no agent will again change its id in all subsequent steps.

(iv) $u$ and $\upsilon$ obviously know $n$ after their interaction (that terminates the id-assignment process), because $u$, that sets $id=n-1$, sets $ps$ equal to $id+1$, and $\upsilon$ that keeps its id (that is, it still has $id=n-2$), sets $ps=id+2$. At the same time, for all other agents $w\in V-\{u,\upsilon\}$, it holds that their $ps$ variables contain a value less than $n$, because, if not, then there should be an agent other than $u$ with id $n-1$ which is impossible (due to the correctness of the id-assignment process).
\qed
\end{proof}

\begin{lemma} \label{lem:ids2}
After the last id-modification has taken place (via interaction $(u,\upsilon)$), all agents, one after the other, get informed of the correct population size. The propagation of the population size is done clockwise on the virtual ring $n-1,0,1,\ldots,n-2,n-1$. Whenever, after the unique consecutive ids have been correctly assigned, some agent $i$ propagates the population size to agent $i+1$, $i$ becomes reinitialized and cannot become reinitialized again in the future.
\end{lemma}
\begin{proof}
During the last id-modification, $u$ increases its id to $n-1$ while, on the other hand, the responder $\upsilon$ keeps its id to $n-2$. Now, all agents have been assigned the correct unique ids $\{0,1,\ldots,n-1\}$ and, according to Lemma \ref{lem:ids1}, no agent can again alter the value of its $id$ variable. We can, thus, from now on call the agents by their ids. Agent $n-1$ (which is  agent $u$, with id $n-1$) sets both variables $search\_for$ and $ssearch\_for$ to $0$ meaning that it may now only interact with agent $0$; all other interactions that may happen until it meets agent $0$, simply have no effect. Agent $n-2$ (that is, $\upsilon$) sets those variables to $-2$, that is, it is \emph{waiting} for a specific interaction (which will be explained later on). Finally, both set $ps$ and $sps$ to the correct population size $n$, restore $\mathcal{A}$'s working blocks to the respective input strings (by executing $working\leftarrow binput$), and clear their output tapes (denoted by $output\leftarrow \emptyset$ in the code).

Moreover, from now on, an agent with id $k$ other than $n-1$ and $n-2$ can only set $ps$ to $n$ if it interacts with agent $k-1$ (if $k=0$, $k-1$ is defined to be $n-1$) and if $k-1$ already knows that the population size is $n$. As we have already shown in Lemma \ref{lem:ids1}, only agents $n-2$ and $n-1$ know that the population size is $n$. All interactions between agent $n-2$ and some other agent $w$ have no effect, except for the case in which $w$'s id is $n-3$ and it has already obtained $n$. Thus, $n-2$ cannot inform any agent of the population size and $n$ cannot be propagated \emph{counterclockwise}. Moreover, agent $n-1$ can only propagate $n$ to agent $0$, $0$ to $1$, and so on, which implies that the agents can only learn $n$ via \emph{clockwise} propagation. We next show that this is what eventually happens.

Eventually, in a fair execution, agent $n-1$ will interact with agent $0$. When this happens, agent $n-1$ has just found the agent that it was searching for, thus it sets $search\_for$ and $ssearch\_for$ to $-1$ meaning that from now on it is not searching for any agent and executes $\mathcal{A}$ for $c=\mathcal{O}(1)$ steps while ignoring any received data. The reason for executing $\mathcal{A}$ for a fixed number of steps will be explained in the sequel. What is of interest here, is that $\mathcal{I}$ now, in fact, has correctly reinitialized $\mathcal{A}$'s execution on agent $n-1$, as if the whole id-assignment process had never been executed. The reason is that $\mathcal{A}$ can now read the correct $id$ and the correct population size, the working block simply contains the agent's input string, the output tape is empty, and all received data during the interaction have been ignored. Now, agent $0$ realizes that the other agent was searching for it (i.e. $rsearch\_for==id$ is \emph{true}) so it \emph{resets itself} (executes $working\leftarrow binput$ and $output\leftarrow \emptyset$) and, since $rps>ps$ is \emph{true} (because, as we have already shown, all agents except for $n-2$ and $n-1$ still know a population size that is strictly smaller than $n$), it sets $search\_for$ and $ssearch\_for$ to $1$ (now searching for agent $1$ to forward its knowledge). Finally, it sets $ps$ and $sps$ to the correct population size, received from agent's $n-1$ message. It is easy to see that this process is continued until agent $n-2$ interacts with agent $n-3$ and the latter knows the correct population size. At this point, both set $search\_for$ and $ssearch\_for$ to $-1$ and both execute $\mathcal{A}$ for $c=\mathcal{O}(1)$ steps while ignoring any received data, thus now all agents have been correctly reinitialized. Finally, it is easy to see that no agent may be reinitialized again in future steps. The reason is that ids won't change and, moreover, no agent will alter its $search\_for$ variable (all agents will forever keep it to $-1$), thus, the code lines that reset $\mathcal{A}$'s data cannot be executed again in the future.
\qed
\end{proof}

In the following Lemma, we call an interaction \emph{effective} if protocol $\mathcal{A}$ is executed in at least one of the participating agents.
\begin{lemma} \label{lem:ids3}
After the unique consecutive ids $\{0,1,\ldots,n-1\}$ have been assigned, agents that have reinitialized themselves can only have \emph{effective interactions} with each other, that is, when they interact with each other $\mathcal{A}$ is executed, while, on the other hand, their interactions with non-reinitialized agents have no effect w.r.t. to $\mathcal{A}$'s execution.
\end{lemma}
\begin{proof}
After the unique ids have been successfully assigned, it is like the population is partitioned in two non-communicating classes, the class $R$ of \emph{reinitialized} agents and $N$ of the \emph{non-reinitialized} ones. Initially (just after the unique ids have been successfully assigned), $R=\emptyset$ and $N=\{0,1,\ldots,n-1\}$, that is, all agents are considered as non-reinitialized. An agent $i\neq n-2$ moves from $N$ to $R$ iff it interacts with the next agent $i+1$ on the virtual ring $n-1,0,1,\ldots,n-2,n-1$ and $i$ already knows the correct population size $n$; the only agents that initially know the correct population size are $n-2$ and $n-1$. $n-3$ and $n-2$ will be the last agents to move from $N$ to $R$ as we have already shown in the proof of Lemma \ref{lem:ids2}, and from that point on it will hold that $R=\{0,1,\ldots,n-1\}$ and $N=\emptyset$. Now, for any intermediate step, let $i\in R$ and $j\in N$, where $i\neq j$ and $i,j\in \{0,1,\ldots,n-1\}$ be the interacting agents. There are three cases:
\begin{enumerate}
 \item $i=j+1$ (for $j=n-1$, $i$ is 0): Due to the clockwise propagation of the reinitializations this can only happen if $i=n-1$ and $j=n-2$. But $n-2\in N$ implies that $n-2$ is still waiting (to interact with $n-3$ which is still in $N$), thus, this interaction has no effect.
\item $i=j-1$ (for $j=0$, $i$ is $n-1$): $i\in R$ implies that $j$ already knows the correct population size. But since $j\in N$, $j$ is still searching for $j+1$ to propagate its knowledge, implying that this interaction has also no effect.
\item $i\neq j+1$ and $i\neq j-1$: Assume that an effective interaction takes place. This implies (line $37$ in Protocol \ref{prot:ids}) that both know the same population size and both have $search\_for=-1$. But $i\in R$ which implies that they both know that the population size is $n$. Since $j$ is non-reinitialized but knows the correct population size, either it is agent $n-2$, in which case it has $search\_for=-2$, or it has $search\_for=j+1\neq -1$ (because, due to the clockwise propagation of the population size, it has learned the correct population size from $j-1$ but has not yet found $j+1$ to propagate it). In both cases we have a contradiction because $j$ should have $search\_for=-1$ according to our assumption. Thus, again here, no effective interaction can take place.
\end{enumerate}
Obviously, $\mathcal{A}$ is always executed when two agents in $R$ interact, because they both have $search\_for=-1$ and know the same, correct, population size (to get convinced, inspect $\mathcal{I}$'s code).
\qed
\end{proof}

\begin{lemma} \label{lem:ids4}
No agent ever falls into some infinite loop.
\end{lemma}
\begin{proof}
It is easy to see that $\mathcal{I}$, except for calling $\mathcal{A}$, only performs some variable assignments, which cannot lead into some infinite loop. But while the correct ids have not yet been assigned to the agents, some interacting agents may contain inconsistent data, which could make subroutine $\mathcal{A}$ fall into some infinite loop. By always executing $\mathcal{A}$ for a constant number of steps we guarantee that it won't. In fact, we guarantee that no agent can become busy for an infinite number of steps, because this would disable $\mathcal{I}$'s capability to reinitialize it, if needed.
\qed
\end{proof}

We can now conclude by combining the above lemmata. We have shown that $\mathcal{B}$, by using $\mathcal{I}$, correctly assigns the unique consecutive ids $\{0,1,\ldots,n-1\}$ to the agents (Lemma \ref{lem:ids1}) and informs them of the correct population size (Lemma \ref{lem:ids2}). Then each agent that propagates the population size, to the next one in the clockwise direction, becomes reinitialized, in the sense that it starts executing $\mathcal{A}$ from the beginning and cannot get reinitialized again in future steps (Lemma \ref{lem:ids2}). During this propagation process, $\mathcal{B}$ does not allow non-reinitialized agents, that possibly contain outdated information, to have some effective interaction with reinitialized agents (Lemma \ref{lem:ids3}). Finally, due to the intermittent execution of $\mathcal{A}$, no agent could have ever become busy for an infinite number of steps, thus, it is guaranteed that the reinitializations can always be applied (Lemma \ref{lem:ids4}). Since $\mathcal{A}$ stably computes $p$, the same holds for $\mathcal{B}$, that correctly simulates $\mathcal{A}$, and the theorem follows.
\qed
\end{proof}

We now show that any symmetric predicate in $SPACE(n\log n)$ also belongs to $IPLM$.

\begin{theorem} \label{the:lowiplm}
$SSPACE(n\log n)$ is a subset of $IPLM$.
\end{theorem}
\emph{Proof Sketch}. The IPALOMA model simulates a deterministic TM $\mathcal{M}$ of $\mathcal{O}(n\log n)$ space. Initially the agents shift all the non-blank contents (input strings) to the left in the agent chain $\{0,1,\ldots, n-1\}$, until their collective memory looks like that of $\mathcal{M}$ (input symbols to the left and blank cells to the right in the chain). Then, agent 0 starts simulating $\mathcal{M}$ and when it has no other memory to read to the right it passes $\mathcal{M}$'s control to agent 1, and so on. Whenever $\mathcal{M}$ accepts (or rejects), output $1$ ($0$ resp.) is propagated to all agents.
\qed
\begin{proof}
Let $p:\mathcal{X}\rightarrow \{0,1\}$ be any predicate in $SSPACE(n\log n)$ and $\mathcal{M}$ be the deterministic TM that decides $p$ by using $\mathcal{O}(n\log n)$ space. We construct an IPALOMA protocol $\mathcal{A}$ that stably computes $p$. Let $x$ be any input assignment in $\mathcal{X}$. Each agent receives its input string according to $x$ (e.g. $u$ receives string $x(u)$). The agents have the unique ids $\{0,1,\ldots,n-1\}$ and know the population size $n$. The agent that has the unique id $0$ starts simulating $\mathcal{M}$.

If at some point the transition function of $\mathcal{M}$ moves the head to the right, but the agent's working memory has no other symbol to read (e.g. it reads $\sqcup$), then it writes $1$ and $\mathcal{M}$'s current state to its message tape and becomes ready. When the unique agent with id $1$ interacts with an agent that has $1$ written in its message tape, it starts simulating $\mathcal{M}$ by putting the head over the first symbol in its working tape and assuming that the state of $\mathcal{M}$ is the state that it found written on the message it received. Generally, whenever an agent with id $0\leq i<n-1$ cannot continue simulating $\mathcal{M}$ to the right, it passes control to the agent with id $i+1$. Additionally, if an agent with id $0<j\leq n-1$ that simulates $\mathcal{M}$ ever reaches its leftmost cell and the transition function of $\mathcal{M}$ wants to move the head left, then it informs agent with id $j-1$ to continue the simulation from its last non-blank cell (when they, eventually, interact).

Note now, that at some point $\mathcal{M}$ may want to use its (initially) blank cells to write symbols. This is handled by $\mathcal{A}$ in the following manner. $\mathcal{A}$ first starts using the blank cells of the agent with id $n-1$. If more are needed, it goes to agent $0$, writes a separator to the first blank cell (naturally we can assume that the separator is already there from $\mathcal{A}$'s initialization step, because also $\mathcal{M}$ uses separators, like `,'s, between the different inputs of the symmetric predicate) and starts using those blank cells. Then it can use those of $1,2,\ldots$ and so on until the last blank cell of agent with id $n-2$. In fact, in this approach the distributed memory of $\mathcal{M}$ can be read by beginning from agent $0$ and reading all blocks (the parts of the working memories that are used for the simulation) before the separators until agent $n-2$. Then we read the whole simulation block of agent $n-1$, proceed with the simulation block after the separator of agent $0$ and continue with those blocks (after the separator) of all agents until agent with id $n-2$. Thus, the memory is read in a cyclic fashion.

Another approach would be, to initially transfer (shift) the concatenation of all agents inputs (separated by some symbol, e.g. `,') to agents $0,1,2,\ldots$. Now, the first $k$ agents will contain input data, agent with id $k-1$ will probably also have some blank cells and the remaining agents will contain only blank cells. In this case, the simulation starts again from agent $0$, but now $\mathcal{M}$'s tape can be read sequentially from agent $0$ to agent $n-1$.

Whenever, during the simulation, $\mathcal{M}$ accepts, then $\mathcal{A}$ also accepts; that is, the agent that detects $\mathcal{M}$'s acceptance, writes 1 to its output tape and informs all agents to accept. If $\mathcal{M}$ rejects, it also rejects. Finally, note that $\mathcal{A}$ simulates $\mathcal{M}$ not necessarily on input $x=(s_0,s_1,\ldots,s_{n-1})$ but on some $x^{\prime}$ which is a permutation of $x$. The reason is that agent with id $i$ does not necessarily obtain $s_{i}$ as its input. The crucial remark that completes the proof is that $\mathcal{M}$ accepts $x$ if and only if it accepts $x^{\prime}$, because $p$ is symmetric.
\qed
%Let $s$ be any string in $L$. Obviously, w.l.o.g. $|s|=\mathcal{O}(n\log n)$ because this is the greatest number of %symbols that can be stored in only $\mathcal{O}(n\log n)$ cells. We partition $s$ into $n$ substrings %$s_{0},\ldots,s_{n-1}$, each consisting of $\mathcal{O}(\log n)$ symbols, where $s=s_{0}s_{1}\ldots s_{n-1}$. %Moreover, we add $i$ in binary to the front of each $s_i$ and use a special symbol (e.g. $\$$) as a separator %between $i$ and $s_i$, that is we create $s^{\prime}_{i}=b(i)\$s_i$ for all $i$ ($b(i)$ simply denotes the binary %string representation of $i$).
\end{proof}

\begin{theorem} \label{the:spnlogn}
$SSPACE(n\log n)$ is a subset of $PLM$.
\end{theorem}
\begin{proof}
Follows from $SSPACE(n\log n)\subseteq$ $IPLM$ $=PLM$. (Moreover, by Savitch's theorem \cite{sav}, we have that $SNSPACE(\sqrt{n\log n})$ is a subset of $PLM$.)
\qed
\end{proof}

\removed{

According to Theorem \ref{the:lowiplm}, $SSPACE(n\log n)$ is a subset of $IPLM$ which according to Theorem \ref{the:iplm} is equal to $PLM$. To summarize, $SSPACE(n\log n)\subseteq$ $IPLM=PLM$ (and by Savitch's theorem \cite{sav}, $SNSPACE(\sqrt{n\log n})$ is a subset of $PLM$).

}

\section{An Improved Lower Bound} \label{sec:imlowPLM}

\subsection{The PALOMA Model Simulates Community Protocols} \label{subsec:comm}

Here, we show that the PALOMA model simulates the Community Protocol model. This establishes that $SNSPACE(n\log n)$ is a lower bound for $PLM$, thus, improving that of Theorem \ref{the:spnlogn}. 

\begin{definition}
Let $CP$ denote the class of all symmetric predicates that are stably computable by the community protocol model.
\end{definition}

It was shown in \cite{GR09} that $CP$ is equal to $SNSPACE(n\log n)$.

\begin{definition}
Let $RCP$ denote the class of all symmetric predicates that are stably computable by a restricted version of the community protocol model in which the agents can only have the unique ids $\{0,1,\ldots,n-1\}$.
\end{definition}

We first show that the community protocol model that is restricted in the above fashion is equivalent to the community protocol model.
\begin{lemma}
$RCP=CP$.
\end{lemma}
\begin{proof}
$RCP\subseteq CP$ holds trivially. It remains to show that $CP\subseteq RCP$. Since the community protocol model can only perform comparisons on ids, it follows that if we replace any vector of unique ids $(id_0,id_1,\ldots, id_{n-1})$ indexed by agents, where $id_0<id_1<\ldots<id_{n-1}$, by the unique ids $(0,1,\ldots,n-1)$ (thus preserving the ordering of the agents w.r.t. their ids) then the resulting computations in both cases must be identical.
\qed
\end{proof}

\begin{lemma}
$RCP$ is a subset of $IPLM$.
\end{lemma}
\begin{proof}
PALOMA protocols that already have the unique ids $\{0,1,\ldots,n-1\}$ and know the population size can do whatever community protocols that have the same unique ids can, and additionally can perform operations on ids (they can store them in the agents' memories and perform some internal computation on them).
\qed
\end{proof}

Since, according to Theorem \ref{the:iplm}, $IPLM$ is equal to $PLM$, we have arrived to the following result.

\begin{theorem} \label{the:cp-plm}
$CP$ is a subset of $PLM$.
\end{theorem}
\begin{proof}
Follows from $CP=RCP\subseteq IPLM=PLM$.
\qed
\end{proof}

\begin{theorem} \label{the:comm}
$SNSPACE(n\log n)$ is a subset of $PLM$.
\end{theorem}
\begin{proof}
$SNSPACE(n\log n)$ is a subset of $CP$ \cite{GR09} and then we take into account Theorem \ref{the:cp-plm}.
\qed
\end{proof}

\subsection{The PALOMA Model Directly Simulates a Nondeterministic TM of $\mathcal{O}(n\log n)$ Space} \label{subsec:ntm}

Note that the proof of Theorem \ref{the:comm} depends on the following result of \cite{vEm89}: A Storage Modification Machine can simulate a Turing Machine. The reason is that \cite{GR09} provided an indirect proof of the fact that $SNSPACE(n\log n)$ is a subset of $CP$. In particular, it was proven that Community Protocols can simulate a Storage Modification Machine and then the result of \cite{vEm89} was used to establish that Community Protocols can simulate a nondeterministic TM. Here, and in order to avoid this dependence, we generalize the ideas used in the proof of Theorem \ref{the:iplm} and provide a direct simulation of a nondeterministic TM of $\mathcal{O}(n\log n)$ space by the PALOMA model, thus, providing an alternative proof for Theorem \ref{the:comm}.

\begin{theorem} \label{the:lowPLM}
$SNSPACE(n\log n)$ is a subset of $PLM$.
\end{theorem}
\begin{proof}
By considering Theorem \ref{the:iplm}, it suffices to show that $SNSPACE(n\log n)$ is a subset of $IPLM$. We have already shown that $IPALOMA$ can simulate a deterministic TM $\mathcal{M}$ of $\mathcal{O}(n\log n)$ space (Theorem \ref{the:lowiplm}). We now present some modifications that will allow us to simulate a nondeterministic TM $\mathcal{N}$ of the same memory size. Keep in mind that $\mathcal{N}$ is a decider for some predicate in $SNSPACE(n\log n)$, thus, it always halts. Upon initialization, each agent enters a reject state (writes $0$ to its output tape) and the simulation is carried out as in the case of $\mathcal{M}$.

Whenever a nondeterministic choice has to be made, the corresponding agent gets ready and waits for participating in an interaction. The id of the other participant will provide the nondeterministic choice to be made. One possible implementation of this idea is the following. Since there is a fixed upper bound on the number of nondeterministic choices (independent of the population size), the agents can store them in their memories. Any time a nondeterministic choice has to be made between $k$ candidates the agent assigns the numbers $0,1,\ldots,k-1$ to those candidates and becomes ready for interaction. Assume that the next interaction is with an agent whose id is $i$. Then the nondeterministic choice selected by the agent is the one that has been assigned the number $i \mod k$. Fairness guarantees that, in this manner, all possible paths in the tree representing $\mathcal{N}$'s nondeterministic computation will eventually be followed.

Any time the simulation reaches an accept state, all agents change their output to 1 and the simulation halts. Moreover, any time the simulation reaches a reject state, it is being re-initiated. The correctness of the above procedure is captured by the following two cases.
\begin{enumerate}
\item \emph{If $\mathcal{N}$ rejects then every agent's output stabilizes to $0$}. Upon initialization, each agent's output is $0$ and can only change if $\mathcal{N}$ reaches an accept state. But all branches of $\mathcal{N}$'s computation reject, thus, no accept state is ever reached, and every agent's output forever remains to $0$.
\item \emph{If $\mathcal{N}$ accepts then every agent's output stabilizes to $1$}. Since $\mathcal{N}$ accepts, there is a sequence of configurations $S$, starting from the initial configuration $C$ that leads to a configuration $C^{\prime}$ in which each agent's output is set to $1$ (by simulating directly the branch of $\mathcal{N}$ that accepts). Notice that when an agent sets its output to $1$ it never alters its output tape again, so it suffices to show that the simulation will eventually reach $C^{\prime}$. Assume on the contrary that it doesn't. Since $\mathcal{N}$ always halts the simulation will be at the initial configuration $C$ infinitely many times. Due to fairness, by an easy induction on the configurations of $S$, $C^{\prime}$ will also appear infinitely many times, which leads to a contradiction. Thus the simulation will eventually reach $C^{\prime}$ and the output will stabilize to $1$.
\end{enumerate}
\qed
\end{proof}

Note, also, that we have just provided an alternative way to prove Theorem \ref{the:cp-plm}. It is known \cite{AR07,GR09} that a nondeterministic TM of space $\mathcal{O}(n\log n)$ can simulate the community protocol model. But, according to Theorem \ref{the:lowPLM}, the PALOMA model can simulate such a TM, thus, it can indirectly simulate the community protocol model.

\subsubsection{Simulating Nondeterministic Recognizers of $\mathcal{O}(n\log n)$ Space}

\noindent \\ \\ Here, we generalize the preceding ideas to nondeterministic \emph{recognizers} of $\mathcal{O}(n\log n)$ space. There is a way to stably compute predicates in $SSPACE(n\log n)$ even when the corresponding TM $N$ might loop, by carrying out an approach similar to the one given above. However, since neither an accept nor a reject state may be reached, the simulation is nondeterministically re-initiated at any point that is not in such a state. This choice is also obtained by the nondeterministic interactions. For example, whenever the agent that carries out the simulation interacts with an agent that has an id that is even, the simulation remains unchanged, otherwise it is re-initiated. Notice however that during the simulation, any agent having id $i$ may need to interact with those having neighboring ids, so those must not be able to cause a re-initiation in the simulation.

Correctness of the above procedure is captured by similar arguments to those in the proof of Theorem \ref{the:lowPLM}. If $\mathcal{N}$ never accepts, then no output tape will ever contain a $1$, so the simulation stabilizes to $0$. If $\mathcal{N}$ accepts there is a sequence of configurations $S$, starting from the initial configuration $C$ that leads to a configuration $C^{\prime}$ in which each agent's output is set to $1$. Observe that this is a ``good'' sequence, meaning that no re-initiations take place, and, due to fairness, it will eventually occur.

\section{An Exact Characterization for $PLM$} \label{sec:exact}

We first give an upper bound on $PLM$.
\begin{theorem} \label{the:upPLM}
All predicates in $PLM$ are in the class $NSPACE(n\log n)$
\end{theorem}
\emph{Proof Sketch}. The proof is similar to those that achieve the upper bounds of MPP \cite{CMS09-2} and Community Protocol \cite{AR07}. In particular, it suffices to show that the language corresponding to any predicate stably computable by the PALOMA model can be decided by a nondeterministic TM of $\mathcal{O}(n\log n)$ space. The TM guesses the next configuration and checks whether it has reached one that is output-stable. Note that $\mathcal{O}(n\log n)$ space suffices, because a population configuration consists of $n$ agent configurations each of size $\mathcal{O}(\log n)$.
\qed
\begin{proof}
Let $\mathcal{A}$ be a PALOMA protocol that stably computes such a predicate $p$. A population configuration can be represented as an $n-$place vector storing an agent configuration per place, and thus uses $\mathcal{O}(n\log n)$ space in total. The language $L$ derived from $p$ is the set of such strings that, when each agent receives a single string element, $p$ holds, that is, $L=\{(s_1,s_2,\ldots,s_n)\;|\; s_i\in X \mbox{ for all } i\in\{1,\ldots,n\} \mbox{ and } p(s_1,s_2,\ldots,s_n)=1\}$.

We will now present a nondeterministic Turing Machine $\mathcal{M_A}$ that decides $L$ in $\mathcal{O}(n\log n)$ space. To accept the input (assignment) $x$, $\mathcal{M_A}$ must verify two conditions: That there exists a configuration $C$ reachable from the initial configuration corresponding to $x$ in which the output tape of each agent indicates that $p$ holds, and that there is no configuration $C^{\prime}$ reachable from $C$ under which $p$ is violated for some agent.

The first condition is verified by guessing and checking a sequence of configurations. Starting from the initial configuration, each time $\mathcal{M_A}$ guesses configuration $C_{i+1}$ and verifies that $C_i$ yields $C_{i+1}$. This can be caused either by an agent transition $u$, or an encounter $(u,v)$. In the first case, the verification can be carried out as follows: $\mathcal{M_A}$ guesses an agent $u$ so that $C_i$ and $C_{i+1}$ differ in the configuration of $u$, and that $C_i(u)$ yields $C_{i+1}(u)$.  It then verifies that $C_i$ and $C_{i+1}$ differ in no other agent configurations. Similarly, in the second case $\mathcal{M_A}$ nondeterministically chooses agents $u$, $v$ and verifies that encounter $(u,v)$ leads to $C^{\prime}$ by ensuring that: (a) both agents have their working flags cleared in $C$, (b) the tape exchange takes place in $C^{\prime}$, (c) both agents update their states according to $\gamma$ and set their working flags to $1$ in $C^{\prime}$ and (d) that $C_i$ and $C_{i+1}$ differ in no other agent configurations. In each case, the space needed is $\mathcal{O}(n\log n)$ for storing $C_i$, $C_{i+1}$, plus $\mathcal{O}(\log n)$ extra capacity for ensuring the validity of each agent configuration in $C_{i+1}$.

If the above hold, $\mathcal{M_A}$ replaces $C_i$ with $C_{i+1}$ and repeats this step. Otherwise, $\mathcal{M_A}$ drops $C_{i+1}$. Any time a configuration $C$ is reached in which $p$ holds, $\mathcal{M_A}$ computes the complement of a similar reachability problem: it verifies that there exists no configuration reachable from $C$ in which $p$ is violated. Since $NSPACE$ is closed under complement for all space functions $\geq \log n$ (see Immerman-Szelepcs\'enyi theorem, \cite{Pa94}, pages $151-153$), this condition can also be verified in $\mathcal{O}(n\log n)$ space. Thus, $L$ can be decided in $\mathcal{O}(n\log n)$ space by some nondeterministic Turing Machine, so $L\in NSPACE(n\log n)$.
\qed
\end{proof}

\begin{theorem}
$PLM$ is equal to $SNSPACE(n\log n)$.
\end{theorem}
\begin{proof}
Follows from Theorems \ref{the:comm} (or, equivalently, Theorem \ref{the:lowPLM}), which establishes that $SNSPACE($ $n\log n)\subseteq PLM$, and \ref{the:upPLM}, which establishes that $PLM\subseteq NSPACE(n\log n)$; but for all $p\in PLM$, $p$ is symmetric, thus, $PLM\subseteq SNSPACE(n\log n)$.
\qed
\end{proof}

\section{Conclusions - Future Research Directions} \label{sec:conc}

We proposed the PALOMA model, an extension of the PP model \cite{AADFP06}, in which the agents are communicating TMs of memory whose size is logarithmic in the population size. We focused on studying the computational power of the new model. Although the model preserves uniformity and anonymity, interestingly, we have been able to prove that the agents can \emph{organize themselves into a nondeterministic TM} that makes full use of the agents' total memory (i.e. of $\mathcal{O}(n\log n)$ space). The agents are initially identical, but by executing an \emph{iterative reinitiation process} they are able to assign \emph{unique consecutive ids} to themselves and get informed of the population size. In this manner, we showed that $PLM$, which is the class of predicates stably computable by the PALOMA model, contains all symmetric predicates in $NSPACE(n\log n)$. Finally, by upper bounding $PLM$, we concluded that it is precisely equal to the class consisting of all symmetric predicates in $NSPACE(n\log n)$.

Many interesting questions remain open. Is the PALOMA model \emph{fault-tolerant}? What preconditions are needed in order to achieve satisfactory fault-tolerance? Is it possible for the PALOMA model to simulate the MPP model \cite{CMS09-2}? To prove the latter, it would suffice to show that $NSPACE($ $n\log n)$ is an upper bound for $MP$ (the class of computable predicates by MPP). But we do not expect this to be easy, because it would require to prove that we can \emph{encode} the $\mathcal{O}(n^2)$ sized configurations of MPP by new configurations of $\mathcal{O}(n\log n)$ size whose transition graph is, in some sense, \emph{isomorphic} to the old one (e.g. the new configurations reach the same stable outputs). Finally, $\mathcal{O}(\log n)$ memory per agent seems to behave as a \emph{threshold}. Is there some sort of impossibility result showing that with $\mathcal{O}(f(n))$ memory, where $f(n)$ is asymptotically smaller than $\log n$, the class of stably computable predicates is strictly smaller than $NSPACE(nf(n))$? At a first glance, it seems that the agents are unable to store uids and get informed of the population size.

\end{document}